\documentclass[3p]{elsarticle}

\usepackage[noend]{algpseudocode}
\usepackage{algorithm}
\usepackage{verbatim}
\usepackage{color}
\usepackage[utf8x]{inputenc}
\usepackage{t1enc}
\usepackage{latexsym,amsmath,amssymb}

\usepackage{lineno}
\usepackage{latexsym,amsmath,amssymb}
\usepackage{graphicx,epsfig}
\usepackage{url}
\usepackage{multirow}

\usepackage{datetime}

\newcommand{\Fact}{\textit{Fact}}
\newcommand{\PL}{P\!L}

\def\P{p}

\newcommand{\fmax}{F}
\newcommand{\pr}{P}

\newcommand{\pos}{\mathit{pos}}

\newcommand{\fmin}{f}

\newcommand{\select}{{\textit{select}}}
\newcommand{\rank}{{\textit{rank}}}

\renewcommand{\epsilon}{\varepsilon}

\renewcommand{\epsilon}{\varepsilon}

\newcommand{\PNF}{\mathrm{PNF}}

\newcommand{\LPNone}{{\mathcal L}_{\textrm{PN1}}}
\newcommand{\LPNzero}{{\mathcal L}_{\textrm{PN0}}}
\newcommand{\lext}{\mathcal{L}_{\textrm{ext}}}
\newcommand{\pext}{\textit{ext}}

\newcommand{\pn}{\textit{pnw}}
\newcommand{\crit}{\textit{ecrit}}

\makeatletter
\newcommand{\pushright}[1]{\ifmeasuring@#1\else\omit\hfill$\displaystyle#1$\fi\ignorespaces}
\newcommand{\pushleft}[1]{\ifmeasuring@#1\else\omit$\displaystyle#1$\hfill\fi\ignorespaces}
\makeatother

\newcommand{\floor}[1]{\lfloor #1 \rfloor}

\newtheorem{theorem}{Theorem}

\newtheorem{proposition}[theorem]{Proposition}
\newtheorem{conjecture}[theorem]{Conjecture}
\newtheorem{lemma}[theorem]{Lemma}
\newdefinition{remark}{Remark}
\newdefinition{example}{Example}
\newdefinition{definition}{Definition}
\newproof{proof}{Proof}

\newcommand{\normal}{prefix normal} 
\newcommand{\pword}[1]{$#1$-prefix normal}  % {prefix normal word with respect to $#1$}
\newcommand{\peq}[1]{$#1$-prefix equivalent} % {prefix normal equivalent with respect to $#1$}

\begin{document}

\begin{frontmatter}

\title{On Prefix Normal Words and Prefix Normal Forms\footnote{Preliminary versions of parts of this article have appeared in~\cite{FL11}, \cite{pn1} and~\cite{pn2}.}}

\author[auth1]{P\'eter Burcsi}
\address[auth1]{Dept.\ of Computer Algebra, E\"otv\"os Lor\'and Univ., Budapest, Hungary}
\ead{bupe@compalg.inf.elte.hu}

\author[auth2]{Gabriele Fici}
\address[auth2]{Dip.\ di Matematica e Informatica, University of Palermo, Italy}
\ead{gabriele.fici@unipa.it}

\author[auth3]{Zsuzsanna Lipt\'ak\corref{corrauth}}
\cortext[corrauth]{Corresponding author}
\address[auth3]{Dip.\ di Informatica, University of Verona, Italy}
\ead{zsuzsanna.liptak@univr.it}

\author[auth4]{Frank Ruskey}
\address[auth4]{Dept.\ of Computer Science, University of Victoria, Canada}
\ead{ruskey@cs.uvic.ca}

\author[auth5]{Joe Sawada}
\address[auth5]{School of Computer Science, University of Guelph, Canada}
\ead{jsawada@uoguelph.ca}

\date{}

\begin{abstract}
A \pword{1} word is a binary word with the property that no factor has more $1$s than the prefix of the same length; a \pword{0} word is defined analogously. 
These words arise in the context of indexed binary jumbled pattern matching, where the aim is to decide whether a word has a factor with a given number of $1$s and $0$s (a given Parikh vector).  Each binary word has an associated set of Parikh vectors of the factors of the word.  Using prefix normal words, we provide a characterization of the equivalence class of binary words having the same set of Parikh vectors of their factors.

We prove that the language of prefix normal words is not context-free and is strictly contained in the language of pre-necklaces, which are prefixes of powers of Lyndon words. We give enumeration results on $\pn(n)$, the number of prefix normal words of length $n$, showing that, for sufficiently large $n$,
\[
2^{n-4 \sqrt{n \lg n}} \le \pn(n) \le 2^{n - \lg n + 1}.
\]

For fixed density (number of $1$s), we show that the ordinary generating function of the number of prefix normal words of length $n$ and density $d$ is a rational function. Finally, we give experimental results on $\pn(n)$, discuss further properties, and state open problems.
\end{abstract}

\begin{keyword}
\normal{} words, \normal{} forms, binary languages, binary jumbled pattern matching, pre-necklaces, Lyndon words, enumeration.
\end{keyword}

\end{frontmatter}

%\linenumbers

\section{Introduction}

A binary word is called {\em \pword{1}} if no factor (substring) has more $1$s than the prefix of the same length.
For example, $11010$ is \pword{1}, but $10110$ is not. Similarly, a binary word is called {\em \pword{0}} if no factor has more $0$s than the prefix of the same length. When not further specified, by {\em prefix normal} we mean \pword{1}. In~\cite{pn1}, we gave an algorithm for generating all prefix normal words of fixed length $n$. As we will see later, to each binary word, a \pword{1} word and a \pword{0} word can be associated in a unique way, which we will call its {\em prefix normal forms}. 

\medskip

The {\em Parikh vector} of a binary word $u$ is the pair $(x,y)$, where $x$ is the number of $1$s in $u$, and $y$ is the number of $0$s in $u$. The set of Parikh vectors of factors of a word $w$ is called the \emph{Parikh set} of $w$. For binary words, the problem of deciding whether a particular pair $(x,y)$ lies in the Parikh set of a word $w$ is known as {\em Binary Jumbled Pattern Matching} (BJPM). There has been much interest recently in the indexed version of this problem (IBJPM), where an index for the Parikh set is created in a preprocessing step, which can then be used to answer queries fast. The Parikh set can be represented in linear space due to the following {\em interval property} of binary strings: If $w$ has $k$-length substrings with $x_1$ resp.\ $x_2$ occurrences of $1$, where $x_1< x_2$, then it also has a $k$-length substring with $y$ occurrences of $1$, for every $x_1\leq y \leq x_2$. Thus the Parikh set can be represented by storing, for every $1\leq k \leq |w|$,  the minimum and maximum number of $1$s in a substring of length $k$. Much recent research has focused on how to compute these numbers efficiently~\cite{CFL09,MR10,MR12,CLWY12,BFKL13,GG13,GHLW13}. The problem has also been extended to graphs and trees~\cite{GHLW13,CGGLLRT13}, to the streaming model~\cite{LLZ12}, and to approximate indexes~\cite{CLWY12}. There is also interest in the non-binary variant~\cite{ELP04,CieliebakELSW04,BEL04,CFL09,BCFL12_IJFCS,BCFL12_TOCS,KRR13}, as well as in reconstruction from the Parikh multi-set of a string~\cite{ADMOP10}. 
Applications in computational biology include SNP discovery, alignment, gene clusters, pattern discovery, and mass spectrometry data interpretation~\cite{Boecker07,Benson03,BoeckerJMS08,DuhrkopLMB13,Parida06}.

The current best construction algorithm for the linear size index for IBJPM runs in $O(n^{1.864})$ time~\cite{CL15}, for a word of length $n$. As we will see later, computing the prefix normal forms of a word $w$ is equivalent to creating an index for the Parikh set of $w$. Currently, we know no faster computation algorithms for the prefix normal forms than already exist for the linear-size index. However, should better algorithms be discovered, these would immediately carry over to the problem of IBJPM.
%%%%%%%%%%%%%%%%%%%%%%%%%%%%%%%%%%%%%%%%%%%%%%%%%%%%%%%%%%%%%%%%%%%%%%%%%%%%%%%%%%%%%%%%%%%%%

\medskip

It is worthwhile noting that some relevant sequences have made it into the
On-Line Encyclopedia of Integer Sequences (OEIS~\cite{sloane}):
A194850 is the number of prefix normal words of length $n$,
A238109 is a list of prefix normal words (over the alphabet $\{1,2\}$), and
A238110 is the maximum size of a class of binary words of length $n$ having the same prefix normal form.

\medskip

The paper is organized as follows: Section 2 contains basic definitions and results
about prefix normal words; in particular that there are unique \pword{0} and \pword{1} words
associated with every word, and thus the set of words can be partitioned according
to this association.  In Section 3 we consider the set of prefix normal words, giving several properties and
characterizations and showing that their language is not context free.  One of these properties
is then used in Section 4, which is concerned with counting the number of prefix normal words of
a given length.  Finally, the paper concludes with some open problems in Section 5.

\section{Basics}\label{sec:pnf}
%%%%%%%%%%%%%%%%%%%%%%%%%%%%%%%%%%%%%%%%%%%%%%%%%%%%%%%%%%%%%%%%%%%%%%%%%%%%%%%%%%%%%%%%%%%

A {\em binary word} (or {\em string}) $w=w_1\cdots w_n$ over $\Sigma=\{0,1\}$ is a finite sequence of elements $w_i \in  \Sigma,$ for $i=1,\ldots,n$. Its length $n$ is denoted by $|w|$. We denote by $\Sigma^n$ the set of words over $\Sigma$ of length $n$, by $\Sigma^{*} = \cup_{n\geq 0} \Sigma^n$ the set of finite words  over $\Sigma$, and 
the empty word by $\epsilon$.
Let $w\in \Sigma^{*}$. If $w=uv$ for some $u,v\in\Sigma^{*}$, we say that $u$ is a \emph{prefix} of $w$ and $v$ is a \emph{suffix} of $w$.  A \emph{factor} or \emph{substring} of $w$ is a prefix of a suffix of $w$.  We denote the set of factors of $w$ by $\Fact(w)$.  Let $w=w_1\cdots w_n\in \Sigma^*$, then the word $\tilde{w}=w_{n}w_{n-1}\cdots w_{1}$ is called the {\em reversal} of $w$. A word $w$ s.t.\ $w = \tilde{w}$ is called a {\em palindrome}. 
A {\em binary language} is any subset $\cal L$ of $\Sigma^*$.

We denote by $|w|_{1}$ the number of $1$s in the word $w$; similarly, $|w|_0$ is the number of $0$s in $w$. The {\em Parikh vector} of a word $w$ over $\Sigma$ is defined as $\P(w)=(|w|_{0},|w|_{1})$. The {\em Parikh set} of $w$ is $\Pi(w)=\{\P(v) \mid  v\in \Fact(w)\}$, the set of Parikh vectors of the factors of $w$.
For example $\P(011) = \P(101) = (1,2)$ and $\Pi(011) = \{(0,0),(0,1),(0,2),(1,0),(1,1),(1,2)\} = \Pi(101) \cup \{(0,2)\}$.

Given a binary word $w$, we denote by $\pr_1(w,i)$ the number of $1$s in the prefix of length $i$ and by $\pos_1(w,i)$ the position of the $i$th $1$ in the word $w$, i.e.\ $\pr_{1}(w,i) = |w_1\cdots w_i|_1$ and $\pos_1(w,i) = \min\{ k : |w_1\cdots w_k|_1 = i\}$.  The functions $\pr_0$ and $\pos_0$ are defined similarly.  Note that in the context of succint indexing, these functions are frequently called $\rank$ and $\select$, cf.~\cite{NavMaek07}: We have, for $x=0,1$, $ \pr_x(w,i) = \rank_x(w,i)$ and $ \pos_x(w,i) = \select_x(w,i)$.  

\subsection{Prefix normal words}

\begin{definition}[Maximum-ones and maximum-zeros functions]
Let $w\in \Sigma^*$. We define, for each $0\leq k\leq |w|$:
\begin{eqnarray*}
\fmax_1(w,k) =  \max\{|v|_1 \mid v\in \Fact(w)\cap \Sigma^k\},
\end{eqnarray*}

\noindent the maximum number of $1$s in a factor of $w$ of length $k$.
When no confusion can arise, we also write $F_1(k)$ for $F_1(w,k)$.
The function $\fmax_0(w,k)$ is defined analogously by taking $0$ in place of $1$.
\end{definition}

For a word $w$, we denote by $F_1(w)$ the function $k\mapsto F_1(w,
k)$ (and similarly with other functions taking arguments $w$ and $k$).

\begin{example}\label{ex:1}
 Take $w = 1010011011000111001011$. In Table~\ref{table:val}, we give the values of $\fmax_{1}$ and $\fmax_{0}$ for $w$.

\begin{table}
\centering
\begin{small}
\begin{raggedright}
\begin{tabular}{r || c *{29}{@{\hspace{2.1mm}}l}}
 $k$  &  0\hspace{1ex}  & 1\hspace{1ex} & 2\hspace{1ex} & 3\hspace{1ex} & 4\hspace{1ex} & 5\hspace{1ex} & 6\hspace{1ex} & 7\hspace{1ex} &
8\hspace{1ex} & 9\hspace{1ex} & 10 & 11 & 12 & 13 & 14 & 15 & 16 & 17 & 18 & 19 & 20 &21 &22\\
\hline \rule[-6pt]{0pt}{22pt}
$ \fmax_{1}$ &  0    & 1& 2& 3& 3& 4& 4& 4& 5& 6& 6& 7& 7& 7& 8& 8& 9& 10& 10& 10& 11& 11& 12  \\
$ \fmax_{0}$ &  0    & 1& 2& 3& 3& 3& 4& 4& 5& 5& 6& 6& 7& 7& 7& 8& 8& 9& 9& 10& 10& 10& 10\\
\hline 
\end{tabular}
\end{raggedright}\caption{\label{table:val}The sequences $ \fmax_{1}$ and $ \fmax_{0}$ for the word $w = 1010011011000111001011$.}
\end{small}
\end{table}
\end{example}

\begin{definition}[Prefix normal words]
A word $w\in \{0,1\}^*$ is called {\em \pword{1}} if $\pr_1(w) = \fmax_1(w)$.
It is called {\em \pword{0}} if $\pr_0(w) = \fmax_0(w)$.
In other words, $w$ is \pword{1} (\pword{0}) if and only if it does not have any factors with more $1$s (more $0$s) than the prefix of the same length. When not specified, then by \emph{prefix normal} we mean {\em \pword{1}}.
\end{definition}

\begin{example}  The word $w = 1100110$ is \pword{1}, but the word $w1 = 11001101$ is not \pword{1} because the factor $1101$ has three $1$s,
while the prefix of length $4$ has only two.   Also, $w$ is not \pword{0} since every \pword{0} word, except those of the form $1^*$,  must start with a 0.
\end{example}

We will soon see that it is possible to find, for every word $w$, a \pword{1} word which has the same maximum-ones function $F_1$ as $w$; and analogously for $0$. These will be called the prefix normal forms of $w$. To this end, we define the following equivalence; we will then see that equivalent words have the same prefix normal form.

\begin{definition}[Prefix equivalence]
Two words $v,w\in \Sigma^*$ are called \emph{\peq{1}} if $\fmax_1(v) = \fmax_1(w)$.
They are called \emph{\peq{0}} if $\fmax_0(v) = \fmax_0(w)$.
\end{definition}

\begin{example}
The words $11010, 10110, 01101, 01011$ are all \peq{1}, but not \peq{0}. When considering $0$, we have that $\{01011$, $11010$, $10101\}$ constitute one equivalence class,
and $\{01101, 10110\}$ another one (note that in the first class, there is an additional word not present in the 1-prefix equivalence class).
\end{example}

Next we will show that every equivalence class contains exactly one \normal{} word (Theorem \ref{thm:pnf}), which can thus be used as its representative.  
This will allow us to associate two prefix normal words to every word $w$ (Definition \ref{def:pnf}). First we need the following lemma.

\begin{lemma}\label{lemma:Fa}
Let $w\in \Sigma^*$. Then, for all $0\leq i \leq j \leq |w|$:
$F_1(j) - F_1(i) \leq F_1(j-i)$.
\end{lemma}

\begin{proof}
Observe that if $v = yz$, then $|v|_1 \le F_1(|y|) + F_1(|z|)$.  Thus if $v$ is a length $j$ word such that $|v|_1 = F_1(j)$ and
$|y| = i$, then $F_1(j) \le F_1(i) + F_1(j-i)$. \qed
\end{proof}

\begin{theorem}\label{thm:pnf}
For every $w\in \Sigma^*$ there is a unique \pword{1} word $w'$ such that $\fmax_1(w') = \fmax_1(w)$;
similarly, there is a unique \pword{0} word $w''$ such that $\fmax_0(w'') = \fmax_0(w)$.
\end{theorem}

\begin{proof} We only give the proof for $w'$. The construction of $w''$ is analogous.

First note that if the \pword{1} words $u$ and $v$ are \peq{1}, then necessarily $u=v$. This holds because the prefix function $P_1$ determines the word, i.e.\ $P_1(u) = P_1(v)$ implies $u=v$ for {\em any} $u,v$. But since $u$ and $v$ are \pword{1} words, their prefix and maximum-ones functions coincide, and since they 
are \peq{1}, we have $P_1(u) = F_1(u) = F_1(v) = P_1(v)$. This proves uniqueness. 

Next, we will construct $w'$, given $w$. It is easy to see that for $1\le k\le |w|$, one has either $ \fmax_1(w,k)= \fmax_1(w,k-1)$ or $ \fmax_1(w,k)=1+ \fmax_1(w,k-1)$.
Now define the word $w'$ by

 $$w'_k = \begin{cases}
1 & \quad \text{if $ \fmax_1(w,k)=1+ \fmax_1(w,k-1)$}\\
0 & \quad \text{if $ \fmax_1(w,k)= \fmax_1(w,k-1)$}
\end{cases}$$
for every $1\le k\le |w|$.

By construction, we have $\pr_{1}(w',k) = \fmax_{1}(w,k)$ for every $1\le k\le |w|$. We still need to show that $\pr_1(w',k) = \fmax_1(w',k)$ for all $k$. This will prove that $w'$ is \pword{1}, as well as that it is \peq{1} to $w$. 

By definition, $\pr_1(w',k) \leq \fmax_1(w',k)$ for all $k$. Now let $v\in Fact(w')$, $|v|=k$, and $v = w_{i+1}\cdots w_j$. Then $|v|_1 = \pr_1(w',j) - \pr_1(w',i) = \fmax_1(w,j) - \fmax_1(w,i) \leq \fmax_1(w,j-i) = \pr_1(w',j-i) = \pr_1(w',k)$, where the inequality holds by Lemma~\ref{lemma:Fa}. 
We have thus proved that $ \fmax_1(w',k) \leq \pr_1(w',k)$, and hence $w'$ is \pword{1}.  
\qed
\end{proof}

\subsection{Normal forms and Parikh sets}

\begin{definition}[(Prefix) normal forms]\label{def:pnf}
Let $w\in \Sigma^*$. Then we denote by $\PNF_1(w)$ the unique \pword{1} word which is \peq{1} to $w$, and by $\PNF_0(w)$ the unique \pword{0} word which is \peq{0} to $w$. We refer to $\PNF_1(w)$ and $\PNF_0(w)$ as the {\em prefix normal form w.r.t.\ $1$ (resp.\ w.r.t.\ $0$)} or just {\em normal form w.r.t.\ $1$ (resp.\ w.r.t.\ $0$)} of $w$.
\end{definition}

\begin{example}
 Let  $w = 1010011011000111001011$.  The normal forms of $w$ are the words $$\PNF_{1}(w)=1110100110100101100101,$$  $$\PNF_{0}(w)=0001101010101101010111.$$
 
 \noindent Refer to Example~\ref{ex:1} for the values of the two functions $F_1(w)$ and $F_0(w)$. 
 \end{example}

\medskip

The operators $\PNF_1$ and $\PNF_0$ are idempotent operators; i.e., if $u = \PNF_x(w)$ then $\PNF_x(u) = u$, for $x=0,1$. 
This gives us an equivalent definition of prefix normality: a word $w$ is \pword{x} if $\PNF_x(w) = w$.  Also, for any $w\in \Sigma^{*}$ and $x\in \Sigma$, it holds that $\PNF_{x}(w)=\PNF_{x}(\tilde{w})$. Note further that if the equivalence class of $w$ contains only one element, then $w$ is necessarily prefix normal and a palindrome. In Table~\ref{table:classes4} we list all eight 1-prefix equivalence classes for words of length $4$.

\begin{table}[ht]
\begin{center}
\begin{small}
\begin{raggedright}
\begin{tabular}{l *{3}{@{\hspace{6mm}}c}}
$\PNF_{1}$   & Class  & Cardinality\\
\hline \rule[-2pt]{0pt}{3pt}\\
$1111$ & \{$1111$\} & 1\\
$1110$ & \{$1110$, $0111$\}& 2\\
$1101$ & \{$1101$, $1011$\}& 2\\
$1100$ & \{$1100$, $0110$, $0011$\}& 3\\
$1010$ & \{$1010$, $0101$\}& 2\\
$1001$ & \{$1001$\}& 1\\
$1000$ & \{$1000$, $0100$, $0010$, $0001$\}& 4\\
$0000$ & \{$0000$\}& 1\\
\hline \vspace{4mm}
\end{tabular}
\end{raggedright}
\caption{The sets of \peq{1} words of length $4$.\label{table:classes4}}
\end{small}
\end{center}
\end{table}

\medskip

The normal forms of a word allow us to determine the Parikh vectors of the factors of the word, as we will show in Theorem \ref{thm:charparset}. We first recall the following lemma from~\cite{CFL09} (which also appears to be folklore). We say that a Parikh vector $q$ {\em occurs} in a word $w$ if $w$ has a factor $v$ with $p(v) = q$.

\begin{lemma}[Interval Lemma~\cite{CFL09}]
\label{lemma:interval}
Let $w\in \Sigma^*$. Fix $1\leq k\leq |w|$.
If the Parikh vectors $(x_1,k-x_1)$ and $(x_2,k-x_2)$ both occur in $w$, then so does $(y,k-y)$ for any $x_1\leq y \leq x_2$.
\end{lemma}

The lemma can be proved with a simple sliding window argument, exploiting the fact that when a fixed size window is shifted by one, then the number of $1$s in the window changes by at most one.

\begin{theorem}\label{thm:charparset}
Let $w,w'$ be words over $\Sigma$. Then $\Pi(w) = \Pi(w')$ if and only if $\PNF_1(w) = \PNF_1(w')$ and $\PNF_0(w) = \PNF_0(w')$.
\end{theorem}

\begin{proof}
Let $\fmin_1(w,k)$ denote the minimum number of $1$s in a factor of $w$ of
length $k$.
As a direct consequence of Lemma~\ref{lemma:interval}, we have that
for a Parikh vector $q=(x,y)$, $q\in \Pi(w)$ if and only if
$\fmin_1(w,x+y) \leq x \leq \fmax_1(w,x+y)$. Thus  for two words
$w,w'$, we have $\Pi(w) = \Pi(w')$ if and only if $\fmax_1(w) =
\fmax_1(w')$ and $\fmin_1(w) = \fmin_1(w')$.  It is easy to see that
for all $k$, $\fmin_1(w,k) = k - \fmax_0(w,k)$, thus the last
statement is equivalent to $\fmax_1(w) = \fmax_1(w')$ and $\fmax_0(w)
= \fmax_0(w')$. This holds if and only if $\PNF_1(w) = \PNF_1(w')$ and
$\PNF_0(w) = \PNF_0(w')$, and the claim is proved. \qed
\end{proof}

Define $I(w) = \{ (P_0(w,k),P_1(w,k)) \mid 0 \le k \le |w| \}$, the set of Parikh vectors
of all prefixes of $w$.  The following lemma is immediate.

\begin{lemma}
For all $w \in \Sigma^*$,
\[
\Pi(w) = \bigcup_{i=1}^{n} I(w_i \cdots w_n).
\]
\label{lemma:suffix}
\end{lemma}

There is an interesting geometrical way to view Lemma \ref{lemma:suffix} which we describe now.
Imagine each Parikh pair as the coordinates of a point in the Euclidean plane that has been rotated clockwise $\pi/4$ radians.
Each word $w$ can be interpreted as a polygonal path in this plane going up and to the right for each 1 ($\nearrow$) or
down and to the right for each 0 ($\searrow$), for each successive bit of $w$.  To obtain $\Pi(w)$ imagine grabbing
the polygonal path for $w$ and pulling it one step at a time through the origin, keeping track of the integer lattice
points that are hit after each pull (and ignoring the stuff to the left of the origin).
The normal forms $\PNF_{1}(w)$ and $\PNF_{0}(w)$ are obtained by forming polygonal paths starting at the origin, and
connecting the uppermost and the lowermost points of the region, respectively.

\vspace{-.2cm}
\begin{figure}[h]
\begin{center}
  \includegraphics[height=45mm]{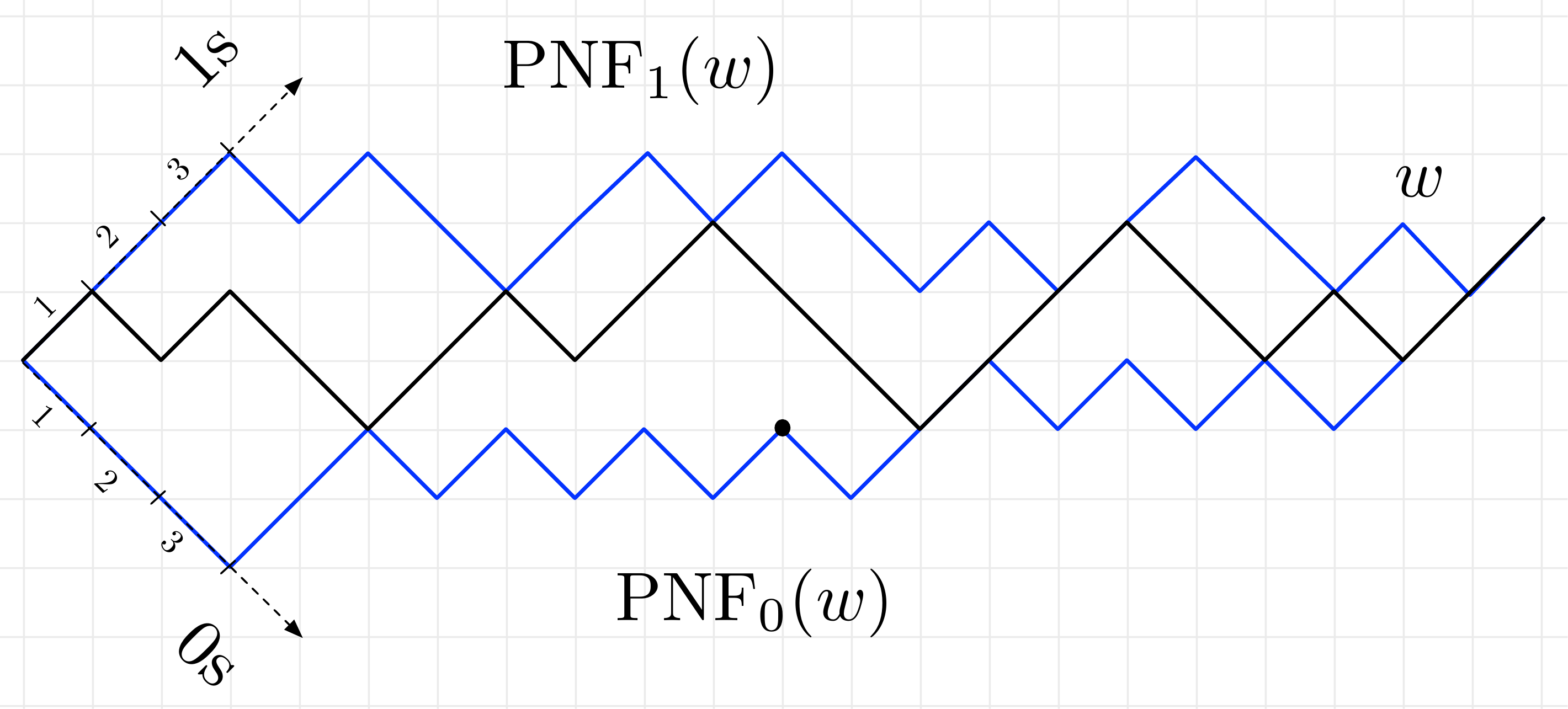}
  \caption{The word $w=1010011011000111001011$ (dark line), its normal forms $\PNF_{1}(w)=1110100110100101100101$ and $\PNF_{0}(w)=0001101010101101010111$ (lighter lines); the region between the two is the Parikh set of $w$; e.g.\  $w$ has a substring containing $5$ ones and $6$ zeros (black dot). Note that the axes giving the number of $0$s and $1$s are rotated by 45 degrees clockwise.}\label{fig:esempio}
\end{center}
\end{figure}

\subsection{Indexing for binary jumbled pattern matching}

Theorem~\ref{thm:charparset} is relevant for the problem known as Indexed Binary Jumbled Pattern Matching, which has attracted much interest recently. Recall that a Parikh vector over $\{0,1\}$ is a multiplicity vector of a string, i.e.\ it has non-negative integer entries. 

\begin{quote}{\sc Indexed Binary Jumbled Pattern Matching (IBJPM)}\\
Given a  string $w$ of length $n$ over $\{0,1\}$, create an index which allows fast answers to queries of the following form: \\
\textbf{Input:} a Parikh vector $q$, \\
\textbf{Output:} return {\bf yes} if $q$ occurs in $\Pi(w)$, and {\bf no} otherwise.
\end{quote}

For $1\leq k\leq n$, let $\fmin_1(w,k)$ be the minimum number of $1$s in a factor of length $k$, and $\fmax_1(w,k)$, as before, the maximum number of $1$s in a factor of length $k$. It follows from Lemma~\ref{lemma:interval} that the answer for query $q=(x,y)$ is {\bf yes} if and only if $\fmax_1(w,x+y) \geq x \geq \fmin_1(w,x+y)$. Therefore, it suffices to store, for every $1\leq k \leq n$, the two numbers $\fmax_1(w,k)$ and $\fmin_1(w,k)$, and queries can be answered in constant time. The size of this data structure is $O(n)$.

All current solutions for IBJPM are based on this observation. The crux is how to construct this linear size data structure. The construction time of the index has steadily decreased since its first introduction: from $O(n^2)$~\cite{CFL09} to $O(n^2/\log n)$~\cite{BCFL10,MR10}, to $O(n^2/\log^2 n)$ in the word RAM-model~\cite{MR12}, to $n^2/2^{\Omega(\log n / \log\log n)^{1/2}}$~\cite{HLRW14}. The fastest solution at present is due to Chan and Lewenstein and has running time $O(n^{1.859})$~\cite{CL15}.

Normal forms are in effect an encoding of this linear size index. We have already seen that the $\fmax$-function can be viewed as a binary string, namely $\PNF_1(w)$. 
We have observed in the proof of Theorem \ref{thm:charparset} how the function $\fmin_1(w)$ is determined by $F_0(w)$ and thus also by $\PNF_0(w)$, thus we have shown the following lemma.

\begin{lemma}
The answer for an IBJPM query $q=(x,y)$ is {\bf yes} if and only if $\pr_1(\PNF_1(w), x+y) \geq x \geq  \pr_1(\PNF_0(w), x+y)$.
\end{lemma}

Note that $\pr_1$ can be computed in constant time with constant time rank-queries on bit vectors, using only $o(n)$ bits of extra space~\cite{Munro96,Clark96}.

\begin{example}
Let $w=1001101$. Then the linear size data structure is given in the Table~\ref{table:lineards}, and the $F_1$ and $F_0$ functions in Table~\ref{table:lineards2}.

\begin{table}
\begin{center}
\begin{tabular}{c||cccccccc}
$k$ & 0 & 1 & 2 & 3 & 4 & 5 & 6 & 7 \\
\hline
$\fmax_1(w,k)$ &  0    & 1& 2& 2& 3& 3& 3& 4\\
$\fmin_1(w,k)$ &  0    & 0& 0& 1& 2& 2& 3& 4\\
\hline
\end{tabular}
\caption{The maximum and minimum number of $1$s for the word $w=1001101$.\label{table:lineards}}
\end{center}
\end{table}

\begin{table}
\begin{center}
\begin{tabular}{c||cccccccc}
$k$ & 0 & 1 & 2 & 3 & 4 & 5 & 6 & 7 \\
\hline
$F_1(w,k)$ &  0    & 1& 2& 2& 3& 3& 3& 4\\
$F_0(w,k)$ &  0    & 1& 2& 2& 2& 3& 3& 3\\
\hline
\end{tabular}
\caption{The maximum number of $1$s and $0$s for the word $w=1001101$. The normal forms of $w$ are $\PNF_1(w)=1101001$ and $\PNF_0(w) = 0011011$. \label{table:lineards2}}
\end{center}
\end{table}

\end{example}

At present, no faster computation of the normal forms is known than the algorithms cited above for the IBJPM problem. But the connection shown here implies that, should a fast normal form computation be found, it would immediately translate into a new solution for IBJPM.

\section{The language of prefix normal words}\label{sec:La}
%%%%%%%%%%%%%%%%%%%%%%%%%%%%%%%%%%%%%%%%%%%%%%%%%%%%%%%%%%%%%%%%%%%%%%%%%%%%%%%%%%%%%%%%%%%%%

In this section, we take a closer look at prefix normal words. We give several equivalent characterizations of prefix normality, explore some properties of prefix normal words, and then look at the language of prefix normal words. We denote by $\LPNone\subset \Sigma^*$ the language of \pword{1} words, and by $\LPNzero\subset \Sigma^*$ the language of \pword{0} words. Note that these are exactly complemented, i.e.\ replacing every $1$ by a $0$ and vice versa, in each word of $\LPNone$, yields $\LPNzero$. Therefore, every result about $\LPNone$ has an equivalent formulation for $\LPNzero$, as well. Recall that whenever not further specified, we refer to 1-prefix normality. In Section~\ref{sec:preneck} only, we will talk about \pword{0} words, and we will show that $\LPNzero$ is strictly contained in the language of pre-necklaces, when adopting the usual order $0<1$ on the alphabet.

\subsection{General observations about prefix normal words}

We start with several characterizations of prefix normal words.

\begin{proposition}\label{prop:char}
Let $w\in \Sigma^{*}$. The following properties are equivalent:

\begin{enumerate}
\item $w$ is a prefix normal word;
\item $\forall i,j$ where $0\le i\le j \le |w|$, we have $ \pr_{1}(j) -  \pr_{1}(i) \leq  \pr_{1}(j-i)$;
\item $\forall v\in \Fact(w)$ such that $|v|_{1}=i$, we have $|v|\ge  \pos_1(i)$;
\item $\forall i,j$ such that $i+j-1 \le |w|_{1}$, we have $ \pos_1(i) +  \pos_1(j) -1 \leq  \pos_1(i+j-1)$.
\end{enumerate}
\end{proposition}

\begin{proof}

(1) $\Rightarrow$ (2). Follows from Lemma~\ref{lemma:Fa}, since $\pr_1(w) = \fmax_1(w)$.

(2) $\Rightarrow$ (3). Assume otherwise. Then there exists $v\in \Fact(w)$ s.t.\ $|v| <  \pos_1(k)$, where $k=|v|_1$. Let $v = w_{i+1}\cdots w_j$, thus $j-i=k$. Then $\pr_1(j) - \pr_1(i) = k$. But $\pr_1(j-i) = \pr_1(|v|) \leq k-1 < k = \pr_1(j) - \pr_1(i)$, a contradiction.

(3) $\Rightarrow$ (4). Again assume that the claim does not hold. Then there are $i,j$ s.t.\ $ \pos_1(i+j-1) <  \pos_1(i) +  \pos_1(j) -1$. Let $k= \pos_1(j)$ and $l =  \pos_1(i+j-1)$ and define $v = w_k\cdots w_l$. Then $v$ has $i$ many $1$s. But $|v| =  \pos_1(i+j-1) -  \pos_1(j) + 1 <  \pos_1(i) +  \pos_1(j) - 1 -  \pos_1(j) + 1 =  \pos_1(i)$, in contradiction to (3).

(4) $\Rightarrow$ (1). Let $v\in \Fact(w)$, $|v|_1 = i$. We have to show that $\pr_1(|v|) \geq i$. This is equivalent to showing that $ \pos_1(i) \leq |v|$. Let $v = w_{l+1} \cdots w_r$, thus $\pr_1(r) - \pr_1(l) = i$. Let $j = \pr_1(l)+1$, thus the first $1$ in $v$ is the $j$'th $1$ of $w$. Note that we have $l<  \pos_1(j)$ and  $r\geq  \pos_1(i+j-1)$. By the assumption, we have $ \pos_1(i) \leq  \pos_1(i+j-1) -  \pos_1(j) + 1 \leq r-l = |v|$. \qed
\end{proof}

Next we formulate a characterization of the prefix normal property that will be useful in the enumeration of fixed-length prefix normal words (Section~\ref{sec:enumeration}).

\begin{lemma}
\label{lemma:equivDefPn}
Let $w\in 1\Sigma^*$. For some sequence of positive integers $r_1$, $r_2$, $\ldots$, $r_{d-1}$ we can write $w=10^{r_1-1}10^{r_2-1}\cdots 10^{r_d-1}$.
The word $w$ is prefix normal if and only if the following inequalities hold.
\begin{equation*}
\begin{array}{rcll}
r_1 &\leq& r_j \qquad & j=2, 3, \ldots, d-1 \\
r_1+r_2 &\leq& r_j+r_{j+1} \qquad & j=2, 3, \ldots , d-2 \\
&\vdots&& \vdots\\
r_1+r_2+\cdots + r_{d-2} &\leq& r_j+r_{j+1}+\cdots + r_{d-1} \qquad & j=2
\end{array}
\end{equation*}
\end{lemma}

\begin{proof}
Note that for $k=1, 2, \ldots d-1$, we have $\pos_1(k) = 1+\sum_{j=1}^{k-1} r_j$. The statement of the lemma then follows by property (4) of Proposition \ref{prop:char}.
\qed
\end{proof}

We now give some simple facts about the language $\LPNone$.

\begin{proposition}\label{prop:La_properties} Let $\LPNone$ be the language of prefix normal words. \hfill \phantom{.}
\begin{enumerate}
\item $\LPNone$ is prefix-closed, that is, any prefix of a word in $\LPNone$ is a word in $\LPNone$.
\item If $w\in \LPNone$, then any word of the form $1^{k}w$ or $w0^{k}$, $k\ge 0$, also belongs to $\LPNone$.
\item Let $|w|_1 <3$. Then $w\in \LPNone$ iff either $w=0^n$ for some $n\ge 0$ or the first letter of $w$ is $1$.
\item Let $w\in \Sigma^{*}$. Then there exist infinitely many $v\in \Sigma^{*}$ such that $vw\in \LPNone$.
\end{enumerate}
\end{proposition}

\begin{proof}
The claims {\em 1., 2., 3.} follow easily from the definition. For {\em 4.},  note that for any $n\ge |w|$, the word $1^{n}w$ belongs to $\LPNone$. \qed
\end{proof}

We now deal with the question of how a prefix normal word can be extended to the right into another prefix normal word.

\begin{lemma}\label{lemma:test}
Let $w\in \LPNone$. Then $w1\in \LPNone$ if and only if for every $0\le k<|w|$ the suffix of $w$ of length $k$ has less $1$s than the prefix of $w$ of length $k+1$.
\end{lemma}

\begin{proof} Note that for all $1\leq k\leq |w|$, $P_1(w1,k)=P_1(w,k)$. Now if $w1\in \LPNone$, then for the $k$-length suffix $u$ of $w$: $|u|_1 < |u1|_1 \leq P_1(w1,k+1) = P_1(w,k+1)$. Conversely, let $u$ be a factor of $w1$. If  $u$ is a factor of $w$, then $|u|_1\leq P_1(w,|u|)=P_1(w1,|u|)$. Else $u=u'1$, with $u'$ a suffix of $w$, and $|u|_1 = |u'|_1+1 < P_1(w,|u'|+1)+1 = P_1(w1,|u|)+1 = P_1(w1,|u|)+1$, and thus $|u|_1 \leq P_1(w1,|u|)$. Therefore, $w1\in\LPNone$.
\qed
\end{proof}

We close this section by proving that $\LPNone$ is not context-free.

\begin{theorem}\label{teor:CF}
$\LPNone$ is not context-free.
\end{theorem}

\begin{proof}
Recall that the intersection of a CFL with a regular language is a CFL.
We will show that $L' = \LPNone \cap 1^* 0 1^* 0 1^*$ is not a CFL by using the pumping lemma.
Let $n$ be the constant of the pumping lemma and let $z = 1^n 0 1^n 0 1^n \in L'$.
Let $z = uvwxy$ be the usual factorization of the pumping lemma, where we may assume that
$|vx| \ge 1$, $|vwx| \le n$, and for all $i \ge 0$ we have $uv^iwx^iy \in L'$.
Clearly $vx$ cannot contain 0s.  If $vx$ contains some 1s from the first block of 1s in $z$, then
taking $i = 0$  give a contradiction since the third block of 1s is too long.
If $vx$ contains no 1s from the first block of 1s then taking $i = 2$ makes the second or
third block of 1s too long.  \qed
\end{proof}

%%%%%%%%%%%%%%%%%%%%%%%%%%%%%%%%%%%%%%%%%%%%%%%%%%%%%%%%%%%%%%%%%%%%%%%%%%%%%%%%%%%%%%%%%
\subsection{Connection with Lyndon words and pre-necklaces}\label{sec:preneck}
%%%%%%%%%%%%%%%%%%%%%%%%%%%%%%%%%%%%%%%%%%%%%%%%%%%%%%%%%%%%%%%%%%%%%%%%%%%%%%%%%%%%%%%%%

In this section we explore the relationship between the language $\LPNzero$ of prefix normal words w.r.t.\ $0$ and some known classes of words defined by means of lexicographic properties.  Note that in this section, when referring to prefix normality, we mean with respect to $0$. We assume the usual order $0<1$ on the alphabet.

A \emph{Lyndon word} is a word which is lexicographically strictly smaller than any of its proper non-empty suffixes. Equivalently, $w$ is a Lyndon word if it is the strictly smallest, in the lexicographic order, among its conjugates, i.e., for any factorization $w=uv$, with $u,v$ non-empty words, one has that the word $vu$ is lexicographically greater than $w$ \cite{LothaireAlg}. A word $w$ is a {\em power} if it can be obtained by concatenating two or more copies of another word, i.e.\ if there exists a non-empty $v$ and a $k>1$ such that $w=v^k$. A word that is not a power is called {\em primitive}. Note that, by definition, a Lyndon word is primitive. Let us denote by $Lyn$ the set of Lyndon words over $\Sigma$. One has that $Lyn\not\subseteq  \LPNzero$ and $\LPNzero \not\subseteq Lyn$. For example, the word $w=0101$ belongs to $\LPNzero$ but is not a Lyndon word since it is not primitive. An example of a Lyndon word which is not in normal form is $w=00110100111$.

A necklace is a Lyndon word or a power of a Lyndon word. A {\em pre-necklace} is a prefix of a necklace~\cite{Ruskey92} (also called  \emph{preprime word} \cite{Knuth42}, or \emph{sesquipower} or \emph{fractional power} of a Lyndon word \cite{ChHaPe04}). Let us denote by $\PL$ the language of pre-necklaces. 
The next proposition shows that every prefix normal word different from a power of the letter $1$ is a prefix of a Lyndon word.

\begin{proposition}\label{prop:preflin}
Let $w\in \LPNzero$ with $|w|_{0}>0$. Then the word $w1^{|w|}$ is a Lyndon word.
\end{proposition}

\begin{proof}
We have to prove that every rotation of $w' = w1^{|w|}$ is strictly greater than $w'$. If the rotation starts at a position within the second half of $w'$, then this is clearly true, since then its first character is $1$, while $w'$ starts with a $0$, $w$ being a prefix normal word containing at least one $0$. So let $v$ be a suffix of $w'$ of length at least $|w|+1$, and let $u$ be the longest common prefix of $v$ and $w'$. If $u=v$, then $v$ is a border (both a prefix and suffix) of $w'$, of length more than half its length, and thus $w'$ has a period of length $i = |w'| - |v| < |w|$, i.e., every character is the same as the one which follows $i$ positions later. Since the second half of $w'$ consists of $1$s only, this implies that so does the first half, contrary to our assumption. So $v$ is not a prefix of $w'$, and therefore $u$ is followed by two different characters in $v$ and in $w'$. Let us write $v=v'1^{|w|}$. If $|u|\geq |v'|$, then $u1$ is a prefix of $v$, implying that $u0$ is a prefix of $w'$, and thus $w'$ is smaller than $v$. If $|u|<|v'|$, assume that $u0$ is a prefix of $v$ and $u1$ of $w'$. Then $w$ has a substring ($u0$) which has more $0$s than the prefix of the same length ($u1$), a contradiction to $w$ being prefix normal. Therefore, again we have that $w'$ is smaller than $v$. 
\qed
\end{proof}

We can now state the following result:

\begin{theorem}\label{theor:PL}
Every prefix normal word is a pre-necklace. 
\end{theorem}

\begin{proof}
If $w$ is of the form $1^{n}$, $n\ge 1$, then $w$ is a power of the Lyndon word $1$, hence it is a pre-necklace. Otherwise, $w$ contains at least one $0$, thus by Proposition \ref{prop:preflin}, it is the prefix of a Lyndon word.  \qed
\end{proof}

The languages $\LPNzero$  and $ \PL$, however, do not coincide. A shortest word in $ \PL$ that does not belong to $\LPNzero$ is $w=00110100$.  Below we give the table of the number of words in $\LPNzero$ of each length $n\leq 16$, compared with that of pre-necklaces. Both sequences are listed in the On-Line Encyclopedia of Integer Sequences~\cite{sloane} (sequences A062692 and A194850), where the reader can find further terms.

\begin{table}[ht]
\centering
\begin{small}
\begin{raggedright}

\begin{tabular}{c || c *{29}{@{\hspace{2.1mm}}l}}
 $n$    & 1\hspace{1ex} & 2\hspace{1ex} & 3\hspace{1ex} & 4\hspace{1ex} & 5\hspace{1ex} & 6\hspace{1ex} & 7\hspace{1ex} &
8\hspace{1ex} & 9\hspace{1ex} & 10 & 11 & 12 & 13 & 14 & 15 & 16 \\
\hline \rule[-6pt]{0pt}{22pt}
$\LPNzero\cap \Sigma^{n}$ & 2 & 3 & 5 & 8 & 14 & 23 & 41 & 70 & 125 & 218 & 395 & 697 & 1273 & 2279 & 4185 & 7568 \\
$ \PL\cap \Sigma^{n}$ & 2 & 3 & 5 & 8 & 14 & 23 & 41 & 71 & 127 & 226 & 412 & 747 & 1377 & 2538 & 4720 & 
8800  \\
\hline 
\end{tabular}
\end{raggedright}
\caption{\label{table:cardofL}The number of words in $\LPNzero$ and in $\PL$ for each length up to 16.}
\end{small}
\end{table}

\section{Enumeration results about prefix normal words}\label{sec:enumeration}

Let $\pn(n)$ denote the number of prefix normal words of length $n$. 
It is an easy consequence both of Lemma~\ref{lemma:equivDefPn} and of Proposition~\ref{prop:La_properties} that $\pn(n)$ grows exponentially. To see this, note that the conditions of Lemma~\ref{lemma:equivDefPn} are always satisfied if $r_1\leq r_2 \leq \ldots \leq r_k$, and thus the number of partitions of $n$ is a lower bound for $\pn(n)$. On the other hand, Proposition~\ref{prop:La_properties} states that for all $w$, $1^{|w|} w$ is prefix normal, so $\pn(2n) \geq 2^n$.

In Table~\ref{table:cardofL}, we give $\pn(n)$ for $n$ up to $16$, the sequence for $n$ up to $50$ can be found in the On-Line Encyclopedia of Integer Sequences~\cite{sloane}, sequence A194850. In Fig.~\ref{fig:pnf_increase} we show the growth ratio for small values of $n$. Two interesting phenomena can be observed: the values seem to approach 2 slowly, i.e., the number of prefix normal words almost doubles as we increase the length by 1. Second, the values show on oscillation pattern between even and odd values.

\begin{figure}[h]
\begin{center}
\includegraphics[width=\textwidth]{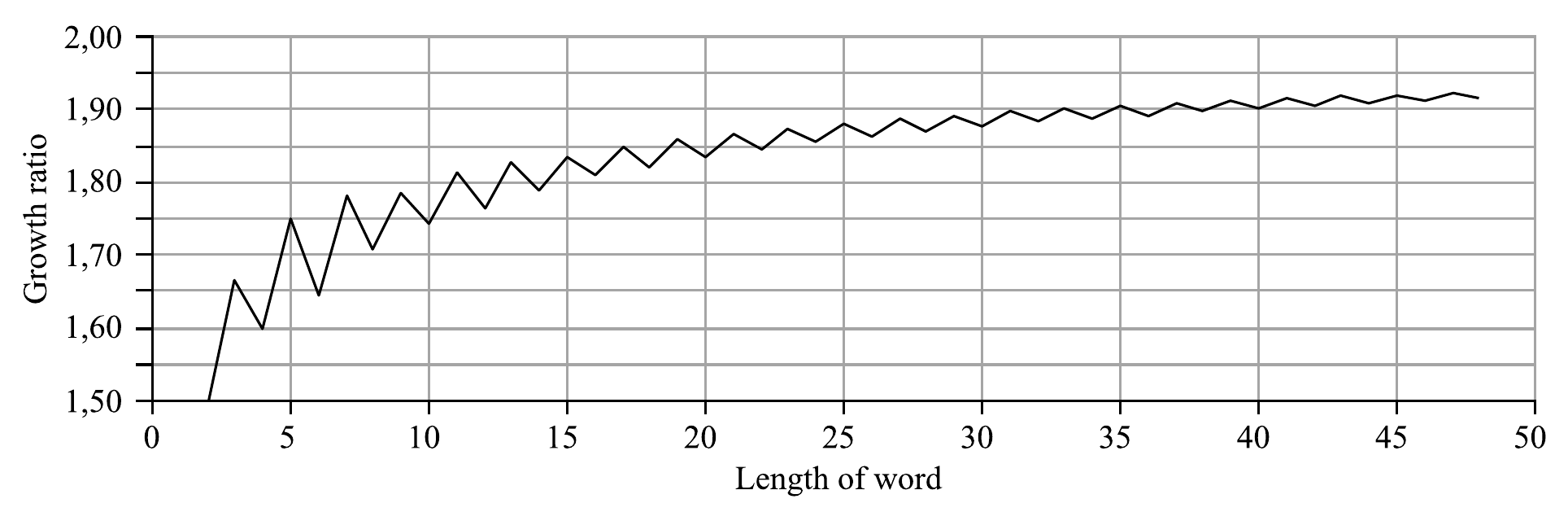}
\end{center}
\caption{The value of $\pn(n)/\pn(n-1)$ for prefix normal words $w$ of length $n$, for $n \leq 50$ (loglinear scale).\label{fig:pnf_increase}}
\end{figure}

\subsection{Asymptotic bounds on the number of prefix normal words}\label{sec:enum1}
%%%%%%%%%

We give lower and upper bounds on the number $\pn(n)$ of prefix normal words of length $n$.

\begin{theorem}

\label{thm:lowBound}
For $n$ sufficiently large
\begin{equation}
\pn(n) \ge 2^{n-4\sqrt{n \log n}}.
\end{equation}
\end{theorem}

\begin{proof}
Let $k=k(n)$ be a positive integer to be fixed later. First we only consider words whose length, $n$, is a multiple of $2k$,
whose first $4k$ letters are $1$s, and in each of the following blocks of length $2k$, there are exactly $k$ $1$s and $k$  $0$s. The number of such words is  $\binom{2k}{k}^{(n-4k)/2k}$ and by construction, they are all prefix normal.

We use the inequality ${2k \choose k} \ge 2^{2k}/(2\sqrt{k})$ and substitute $k = \sqrt{n \log n}$ in the third step.
\begin{align*}
{2k \choose k}^{(n-4k)/2k}
& \ge \left( \frac{2^{2k}}{2\sqrt{k}} \right)^{n/(2k)-2} \\
& = \frac{2^n}{(2\sqrt{k})^{n/(2k)}} \frac{4k}{2^{4k}} \\
& = \frac{2^n}{2^{4\sqrt{n \log n}}} (2\sqrt{k})^{1-n/(2k)} \\
& \ge \frac{2^n}{2^{4\sqrt{n \log n}}}  \text{ for sufficiently large } n.
\end{align*}
The last inequality follows from the fact that $\lim_{n\rightarrow\infty} (2\sqrt{k})^{1-n/(2k)} = 0$ if $k = \sqrt{n \log n}$.
\qed
\end{proof}

Next we show how to obtain an upper bound on $\pn(n)$, considering the length of the first 1-run.

\begin{theorem}
\label{thm:upBound}
For $n$ sufficiently large, we have $\pn(n) \le 2^{n- \lg n + 1}$.
\end{theorem}

\begin{proof}
This will follow from enumeration results about pre-necklaces since every \pword{0} word is a pre-necklace.
Let $PL(n)$ be the number of pre-necklaces of length $n$.  In \cite{Ruskey92} it is shown (top of page 424) that
\[
PL(n) \le \sum_{i=1}^n \frac{2^i}{i} + \sum_{i=1}^n \sqrt{2^i}.
\]
They also show that (Lemma 5 of \cite{Ruskey92})
\[
\lim_{n \rightarrow \infty} \frac{n}{2^n} \sum_{i=1}^n \frac{2^i}{i} = 2.
\]
Thus, for large enough $n$, and fixed $\epsilon > 0$,
\[
PL(n) \le (1+\epsilon) \sum_{i=1}^n \frac{2^i}{i} \le (1+2\epsilon) 2^n/n \le 2^{n- \lg n + 1}.
\]
\hfill \qed

\end{proof}

\subsection{Exact formulas for words with fixed density.}\label{sec:enum2}

For a binary word $w$, its {\em density} is defined as the number of $1$s in $w$, i.e.\ as $|w|_1$. If we count the number of prefix normal words of length $n$ with a given fixed number of $1$s, we get exact results in a few cases. Let us denote by $\pn(n, d)$ the cardinality of the set $\{w \in \LPNone \cap \Sigma^n\mid |w|_1 = d\}$.

\begin{proposition}
\label{prop:genFunc}
For $d=0,1,\ldots,6$, we have the generating functions $f_d(x)=\sum_{n=0}^{\infty} \pn(n,d)x^n$:
\begin{center}
\begin{minipage}{7.2cm}
\begin{eqnarray*}
f_0(x) &=&  \frac{1}{1-x}\\
f_1(x) &=&  \frac{x}{1-x}\\
f_2(x) &=&  \frac{x^2}{(1-x)^2}\\
f_3(x) &=& \frac{x^3}{(1-x^2)(1-x)^2}
\end{eqnarray*}
\end{minipage}
\begin{minipage}{7.2cm}
\begin{eqnarray*}
f_4(x) &=& \frac{x^4}{(1-x^3)(1-x)^3}\\[2mm]
f_5(x) &=&\frac{x^5(1+x+x^2)}{(1-x^4)(1-x^2)^2(1-x)^2}\\[2mm]
f_6(x) &=& \frac{x^6(1+x+x^2+x^3)}{(1-x^5)(1-x^3)(1-x^2)(1-x)^3}
\end{eqnarray*}
\end{minipage}
\end{center}
\end{proposition}

\begin{proof}
For $d\leq3$, one easily checks $\pn(n,0)=\pn(n,1)=1$, $\pn(n,2)=n-1$ and $\pn(n,3) = \floor{(n+1)^2/4}$, giving the desired functions.

For $d=4$, we calculate the number of positive solutions $r_1, r_2, r_3, r_4$ to the inequalities in Lemma \ref{lemma:equivDefPn}.
Let $q_1=r_1-1$, $q_4=r_4-1$, $d_2 = r_2-r_1$ and $d_3=r_3-r_1$. We are counting the nonnegative solutions of
\[
3q_1 + d_2 + d_3 + q_4 + 4 = n,
\]
which give generating function $f_4(x)$ by equating the coefficients of $x^n$ in the expansion of the following product:
\begin{gather}
(1+x^3+x^6+\cdots) (1+x+x^2+\cdots)^3\cdot x^4 \\
= \frac{x^4}{(1-x^3)(1-x)^3}.
\end{gather}
More complicated but manageable case analysis leads to the results for $d=5$ and $6$. \hfill \qed

\end{proof}

Similar formulas can be derived for $\pn(n, n-d)$ for small values of $d$.
Unfortunately, no clear pattern is visible for $f_d(x)$ that we could use for calculating $\pn(n)$.

The inequalities in Lemma \ref{lemma:equivDefPn} define linear diophantine equations. The general theory for enumerating solutions of such equations \cite{stanley} guarantees that there is a closed rational function form for the generating functions with the observed denominators, in \cite{LinDioph} there are algorithms for calculating these functions (which, however are not efficient enough to get results for much larger values of $d$). Above, we only discussed the first few simple cases. We did not succeed in extending our list of concrete formulas for the rational functions $f_d$ for $d>6$ using automated computation.

\subsection{Exact formulas for words with a fixed prefix.}\label{sec:enum3}

We now fix a prefix $w$ and give enumeration results on prefix normal words with prefix $w$. Our first result indicates that we have to consider each $w$ separately.

\begin{definition}
If $w$ is a binary word, let $\lext(w) = \{ w' : ww' \textrm{ is prefix normal} \}$, and $\lext(w, m) = \lext(w)\cap \Sigma^{m}$.
Let $\pext(w, m, d) = |\{ w' : ww'$  is prefix normal of length $|w|+m$  and density $d \}|$, and $\pext(w, m) = |\lext(w, m)|$.
\end{definition}

\begin{lemma}
\label{lemma:extLang}
Let $v, w \in 1\{0,1\}^*$ be both prefix normal. If $v\neq w$ then $\lext(v) \neq \lext(w)$.
\end{lemma}

\begin{proof}
We may assume $|v| \leq |w|$.\\
\emph{First case.} $v$ is not a prefix of $w$. Let $i$ denote the first position where they differ. If $v_i=1$ and $w_i=0$, then for $u=0^{|w|}v$ we have that $vu$ is prefix normal while $wu$ is not. If $v_i=0$ and $w_i=1$, then let $u=0^{|w|}w$. We have that $vu$ is not prefix normal but $wu$ is.\\
\emph{Second case.}
$v$ is a prefix of $w$. If $w$ has a 1 in any position after $|v|$, then we can proceed as in the first case. The remaining case is when $w = v0^m$ for some $m>0$. If $vv$ is prefix normal, then so must be $vvv$, but $v0^mvv$ cannot be. Otherwise, let $k\geq 1$ be the smallest integer (which is sure to exist) such that $v0^kv$ is prefix normal. Then $v0^{k-1}v$ is not prefix normal while $w0^{k-1}v$ is. This completes the proof.
\hfill \qed
\end{proof}

We were unable to prove that the growth of these two extension languages also differ.
\begin{conjecture}
Let $v, w \in 1\{0,1\}^*$ be both prefix normal. If $v\neq w$ then the infinite sequences $(\pext(v, m))_{m\geq1}$ and $(\pext(w,m))_{m\geq 1}$ are different.
\end{conjecture}

The values $\pext(w, m, d)$ seem hard to analyze. We give exact formulas for a few special cases of interest. Using Lemma \ref{lemma:equivDefPn}, it is possible to give formulas similar to those in Proposition \ref{prop:genFunc} for $\pext(w, m, d)$ for fixed $w$ and $d$. We only mention one such result.

\begin{lemma}
For $1\leq d\leq n$ we have $\pext(10, n+d-3, d)   = \pn(n, d)$.
\end{lemma}

\begin{proof}
Consider the following map: let $w$ be an arbitrary word of length $n$ and density $d>1$, starting with $1$. Except for the starting $1$, insert a $0$ right before each subsequent occurrence of 1. This gives a word $w'$ of length $n+d-1$, starting with $10$ that does not contain the factor $11$. Clearly, the map is injective and all words of length $n+d-1$ starting with $10$ and containing no factor $11$ are obtained this way. In order to prove the lemma, we only need to show that prefix normality is preserved by the map and its inverse. For this, observe that there exists a prefix (resp.\ factor) of $w$ of length $k$ containing $r$ $1$s if and only if there exists a prefix (resp.\ factor) of $w'$ of length $k+r-1$ containing $r$ $1$s.
\hfill \qed
\end{proof}

The following lemma lists exact values for $\pext(w, |w|)$ for some infinite families of words $w$. Here $F(n)$ denotes the $n$th Fibonacci number, i.e.\ $F(1)=F(2)=1$ and $F(n+2) = F(n+1)+F(n)$.

\begin{lemma}
For all values of $n$ where the exponents are nonnegative, we have the following formulas:
\begin{center}
\begin{minipage}{7.2cm}
\begin{eqnarray*}
&&\pext(0^n,n) = 1\\[1mm]
&& \pext(1^n, n) = 2^n \\[1mm]
&& \pext(1^{n-1}0, n) = 2^n - 1 \\[1mm]
&& \pext(1^{n-2}01,n) = 2^n - 5 \\[1mm]
&& \pext(1^{n-2}00,n) = 2^n - (n+1)
\end{eqnarray*}
\end{minipage}
\begin{minipage}{7.2cm}
\begin{eqnarray*}
&&\pext((10)^{\frac n2}, n) = F(n+2) \textrm { if }n \textrm{ is even }\\[2mm]
&&\pext((10)^{\frac {n-1}2}1,n) = F(n+1) \textrm{ if }n\textrm{ is odd}\\[2mm]
&&\pext(10^{n-2}1,n) = 3\\[2mm]
&&\pext(10^{n-1},n) = n+1
\end{eqnarray*}
\end{minipage}
\end{center}

\end{lemma}

\begin{proof}
For $w=1^n$, $w=1^{n-1}0$, $w=1^{n-2}01$ and $w=1^{n-2}00$, it is easy to count those extensions that fail to give prefix normal words: None for $w=1^n$; only one for $w=1^{n-1}0$, namely $1^{n-1}01^n$; for $w=1^{n-2}01$, those extensions which contain a $1$-run of length $n-1$, namely $1^{n-2}$ followed by any two characters, or $01^{n-1}$; and for $w=1^{n-2}00$, those that contain at least $n-1$ many $1$s in the second half, i.e.\ with second half $1^n, 1^{n-1}0, 1^{n-2}01, \ldots, 01^{n-1}$.

Similarly, for $w=10^{n-2}1$, $w=10^{n-1}$ and $w=0^n$, counting the extensions that yield prefix normal words gives the result in a straightforward way.

Let $n$ be even. For $w=(10)^{\frac n2}$, note that $ww'$ is prefix normal if and only if $w'$ avoids $11$. The number of such words is known to equal $F(n+2)$. For $n$ odd, the argument is similar, with the prefix of interest, $w1$, being of length $n+1$, hence the previous Fibonacci number.
\hfill \qed
\end{proof}

\subsection{Some experimental results about enumeration of prefix normal words}\label{sec:enum4}

We consider extensions of prefix normal words by a single symbol to the right. It turns out that this question has implications for the enumeration of prefix normal words.

\begin{definition}[Extension-critical words]
We call a prefix normal word $w$ {\em extension-critical} if $w1$ is not prefix normal. Let $\crit(n)$ denote the number of extension-critical words in $\LPNone\cap \Sigma^n$.
\end{definition}

The lemma below applies to any family of words $B$ for which $\varepsilon \in B$ and such that $x \in B$ implies $x0 \in B$.

\begin{lemma} For $n\geq 1$ we have
\begin{equation}
\label{eq:crit1}
\pn(n) = 2\pn(n-1) - \crit(n-1) = \pn(n-1)\left(2-\frac{\crit(n-1)}{\pn(n-1)}\right).
\end{equation}
From this it follows that
\begin{equation}
\label{eq:critProduct}
\pn(n) = 2\prod_{i=1}^{n-1} \left(2-\frac{\crit(i)}{\pn(i)}\right).
\end{equation}
\end{lemma}

\begin{proof}
The number of prefix normal words of length $n$ ending in $0$ is $\pn(n-1)$, that of prefix normal words of length $n$ ending in $1$ is $\pn(n-1)-\crit(n-1)$, hence we have  \eqref{eq:crit1}. The product form follows if we use $\pn(n) = \pn(1)  \prod_{i=1}^{n-1} \frac{\pn(i+1)}{\pn(i)}$. \qed
\end{proof}

\begin{lemma} For $n$ going to infinity, $\lim\inf \crit(n)/\pn(n)  =  0$.
\end{lemma}

\begin{proof}
Assume that there exist an integer $N_0$ and a real number $\varepsilon > 0$ such that for $n\geq N_0$ we have $\crit(n)/\pn(n) > \varepsilon$. Then by  \eqref{eq:critProduct} we would have $\pn(n) = O((2-\varepsilon)^n)$, contradicting Theorem \ref{thm:lowBound}. \qed
\end{proof}

We conjecture that in fact the ratio of extension-critical words converges to $0$. We study the behavior of $\crit(n)/\pn(n)$ for $n\leq 49$. The left plot in Fig.~\ref{fig:Crit} shows the ratio of extension-critical words for $n\leq 49$. These data support the conjecture that the ratio tends to $0$. Interestingly, the values decrease monotonically for both odd and even values, but we have $\crit(n+1)/\pn(n+1) > \crit(n)/\pn(n)$ for even $n$. We were unable to find an explanation for this.

The right plot in Fig.~\ref{fig:Crit} shows the ratio of extension-critical words multiplied by $n/\log n$. Apart from a few initial data points, the values for even $n$ increase monotonically and the values for odd $n$ decrease monotonically, and the values for odd $n$ stay above those for even $n$.

\begin{conjecture}
Based on empirical evidence, we conjecture the following:
\begin{eqnarray}
\crit(n) &=& \pn(n) \Theta(\log n / n) ,\\
\pn(n) &=& 2^{n-\Theta((\log n)^2)} .
\end{eqnarray}
\end{conjecture}

Note that the second estimate follows from the first one by  \eqref{eq:critProduct}.

\begin{figure}
\begin{minipage}[c]{6cm}
\begin{center}
\includegraphics[scale=0.3]{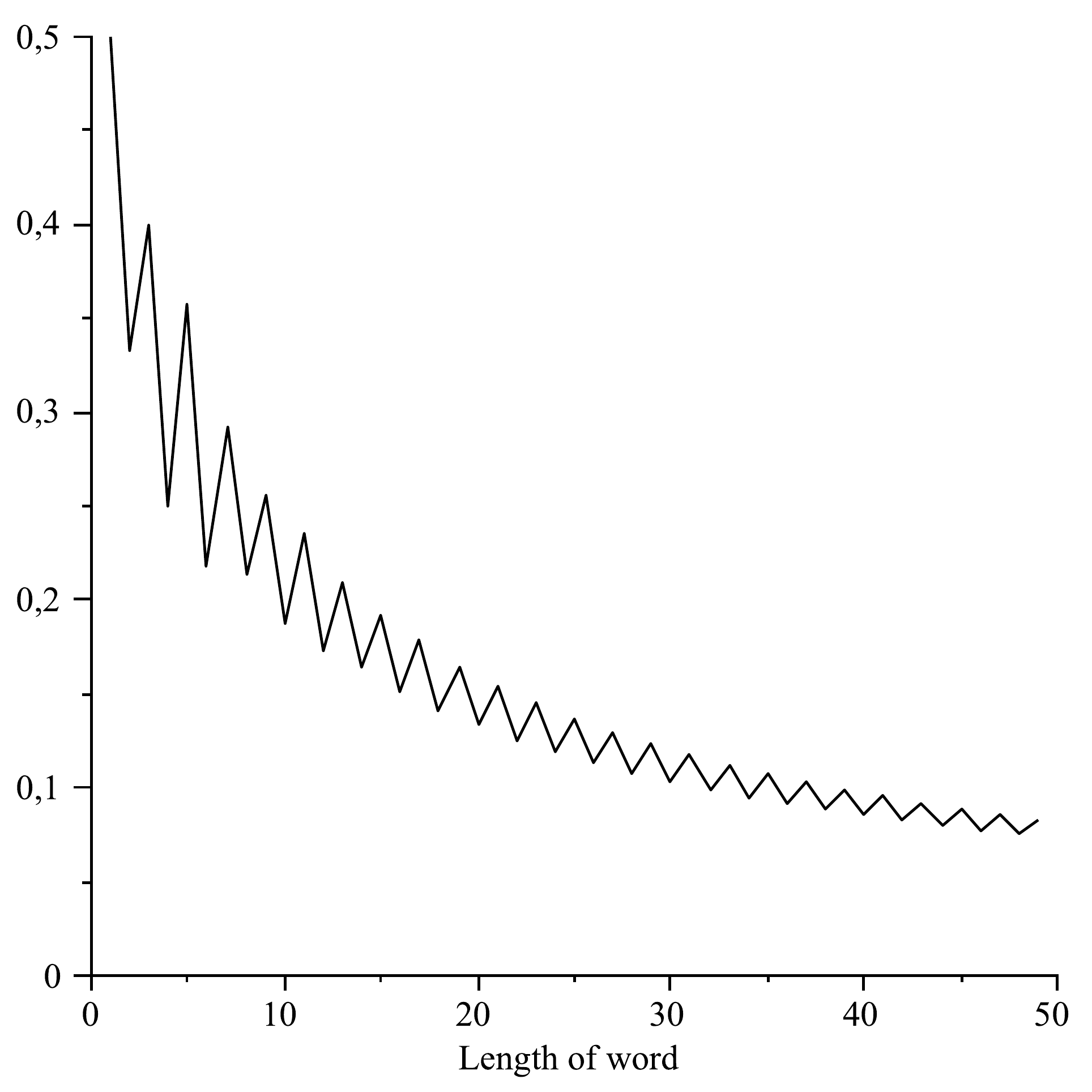}
\end{center}
\end{minipage}
\begin{minipage}[c]{6cm}
\begin{center}
\includegraphics[scale=0.3]{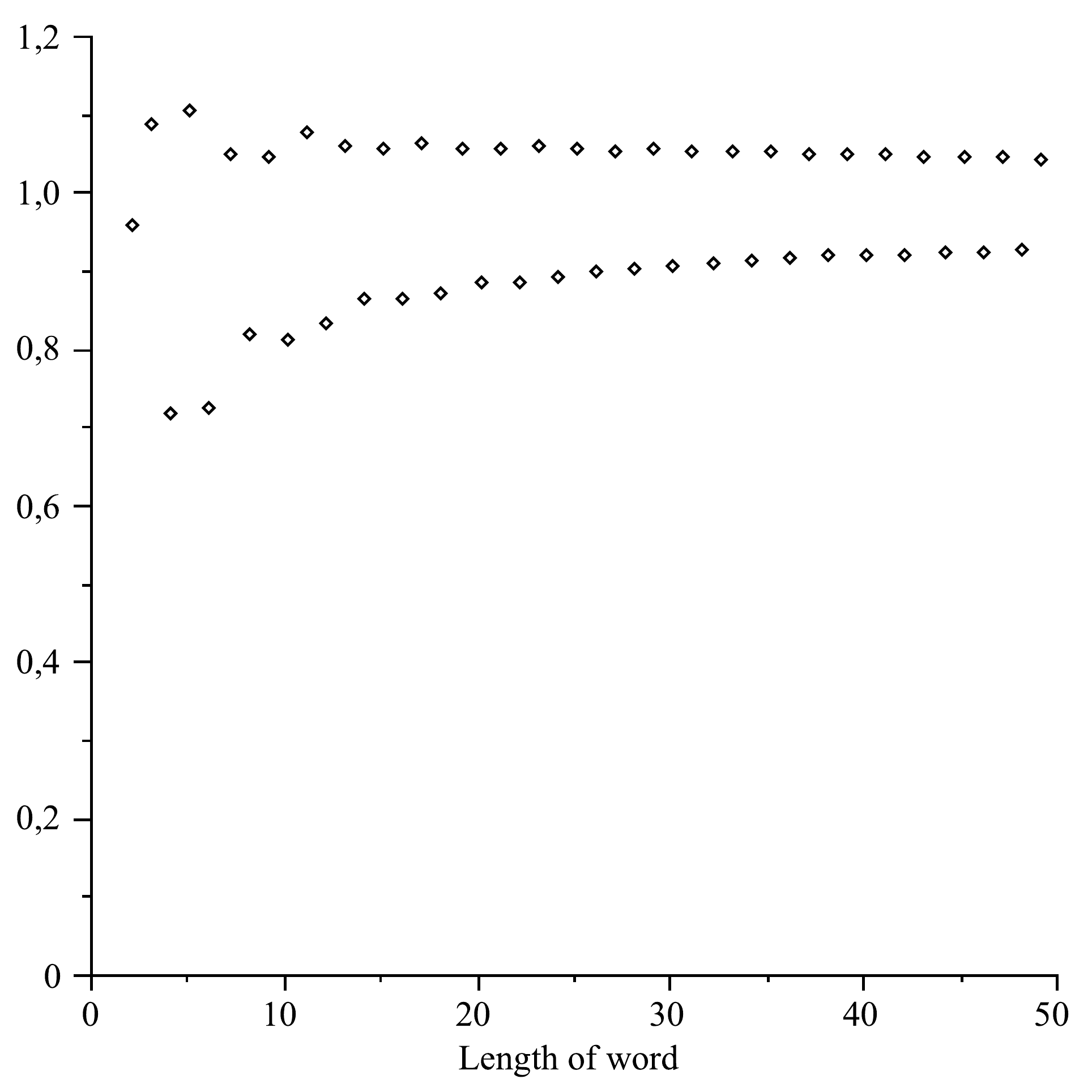}
\end{center}
\end{minipage}
\caption{The ratio $\frac{\crit(n)}{\pn(n)}$ (left), and the value $\frac{\crit(n)}{\pn(n)}\cdot\frac{n}{\ln n}$ (right).}
\label{fig:Crit}
\end{figure}

\section{Conclusion and open problems}
We introduced two new normal forms of binary words, the prefix normal forms with respect to $1$ and $0$, and showed how they arise naturally in the investigation of Parikh sets of binary words and jumbled pattern matching. We introduced prefix normal words (w.r.t.\ $1$ or $0$), words which equal their own normal form, and discussed several properties of these words. We showed results about the language of prefix normal words, among these that \pword{0} are strictly contained in the language of pre-necklaces. We also discussed extensively the growth behavior of the number of fixed-length prefix normal words.

Many open problems remain. It would be nice to have exact, or at least more precise asymptotic formulas for the enumeration of prefix normal words. Related to the enumeration, the strange oscillating behavior in Figures~\ref{fig:pnf_increase} and~\ref{fig:Crit} between odd and even values calls for an explanation.

Another question is testing binary words for prefix normality. Currently, no faster method is known (in worst-case running time), then calculating the normal form.

It would be an interesting direction to explore the connection between the normal forms w.r.t.\ $1$ and $0$, for example how many different values can $\PNF_0(w)$ take (and what can we say about them) if we fix $\PNF_1(w)$.

Finally, prefix normality could also be
defined over non-binary alphabets. In this case however, we do not obtain
an index directly applicable to jumbled pattern matching. Combinatorial or
formal language theoretic investigation and enumeration of prefix
normal words for general alphabets is subject of future work.

\section*{Acknowledgements}
Gabriele Fici was partially supported by the PRIN 2010/2011 project ``Automi e Linguaggi Formali: Aspetti Matematici e Applicativi'' of the Italian Ministry of Education (MIUR).  The research of Joe Sawada and Frank Ruskey was partially funded by grants from the National Engineering Research Council of Canada. We thank an anonymous referee for a very careful reading and helpful suggestions.

%\begin{small}
%\section*{References}
%
%\bibliographystyle{abbrv}
%\bibliography{PrefixNormalTCS}
%
%\end{small}

\end{document}